\documentclass[twoside]{article}
\usepackage[accepted]{aistats2016}

\usepackage{caption}
\usepackage{multicol}
\usepackage[T1]{fontenc}
\usepackage{tgtermes}
\usepackage{amsmath}
\usepackage{scalefnt,letltxmacro}
\LetLtxMacro{\oldtextsc}{\textsc}
\renewcommand{\textsc}[1]{\oldtextsc{\scalefont{1.10}#1}}
\usepackage[scaled=0.92]{PTSans}
\usepackage{amssymb}

\usepackage[ttdefault=true]{AnonymousPro}

\usepackage[english]{babel}
\usepackage[parfill]{parskip}
\usepackage{afterpage}
\usepackage{framed}
\usepackage{xspace}

{\endMakeFramed}

\usepackage{lineno}

\usepackage{ragged2e}

\newcounter{parcount}

\DeclareRobustCommand{\parhead}[1]{\textbf{#1}~}

\usepackage[usenames,dvipsnames]{xcolor}
\definecolor{shadecolor}{gray}{0.9}

\usepackage{graphicx}
\usepackage[labelfont=bf]{caption}
\usepackage[format=hang]{subcaption}

\usepackage{booktabs, array}

\usepackage{natbib}

\usepackage[colorlinks,linktoc=all]{hyperref}
\usepackage[all]{hypcap}
\hypersetup{citecolor=MidnightBlue}
\hypersetup{linkcolor=MidnightBlue}
\hypersetup{urlcolor=MidnightBlue}

\usepackage{cleveref}

\usepackage
[acronym,smallcaps,nowarn,section,nonumberlist]{glossaries}
\glsdisablehyper{}

\newacronym{SGD}{sgd}{stochastic gradient descent}
\newacronym{AISGD}{ai-sgd}{averaged implicit \textsc{sgd}}
\newacronym{ESGD}{\textup{explicit} sgd}{}
\newacronym{ISGD}{implicit sgd}{}
\newacronym{ASGD}{asgd}{}
\newacronym{PROXSAG}{prox-sag}{}
\newacronym{PROXSVRG}{prox-svrg}{}
\newacronym{ADAGRAD}{adagrad}{}

\glsunset{ESGD} 
\glsunset{ISGD}
\glsunset{ASGD}
\glsunset{ADAGRAD}
\glsunset{PROXSAG}
\glsunset{PROXSVRG}

\usepackage{amsmath, amssymb, amsthm, mathtools}
\usepackage{color, enumerate}

\usepackage[plain,noend]{algorithm2e}
\usepackage{graphicx}
\usepackage{caption}
\usepackage{subcaption}
\usepackage{authblk}
\usepackage{tikz}

\newcommand{\nn}{\nonumber}
\newcommand{\eminus}{\text{\sc{e}-}}

\newcommand\defeq{\stackrel{\mathclap{\normalfont\mbox{\scriptsize def}}}{=}}
\newcommand{\commentEq}[1]{
\hspace{10px}\text{ \footnotesize [{\em #1}]}}

\newtheorem{assumption}{Assumption}
\newtheorem{definition}{Definition}
\newtheorem{theorem}{Theorem}

\newtheorem{lemma}{Lemma}
\newtheorem{corollary}{Corollary}

\newcommand{\Reals}[1]{\mathbb{R}^{#1}}

\newcommand{\Ex}[1]{\mathbb{E}\left (#1 \right )} 		 
\newcommand{\ExCond}[2]{\mathbb{E}\left (#1 \right|#2)} 		 
\newcommand{\Var}[1]{\mathrm{Var}\left(#1\right)}   		 

\newcommand{\bigO}[1]{\mathcal{O}(#1)}
\newcommand{\littleO}[1]{o(#1)}

\newcommand{\aisgd}{\gls{AISGD}\xspace}
\newcommand{\sgd}{\gls{SGD}\xspace}
\newcommand{\asgd}{\gls{ASGD}\xspace}
\newcommand{\implicit}{\gls{ISGD}\xspace}
\newcommand{\proxsag}{\gls{PROXSAG}\xspace}
\newcommand{\proxsvrg}{\gls{PROXSVRG}\xspace}
\newcommand{\adagrad}{\gls{ADAGRAD}\xspace}

\newcommand{\loss}{L}
\newcommand{\exloss}{\ell}
\newcommand{\thetaBar}[1]{\bar{\theta}_{#1}}
\newcommand{\thetastar}{\theta_\star}
\newcommand{\thetaim}[1]{\theta^{\mathrm{im}}_{#1}}
\newcommand{\mF}[1]{\mathcal{F}_{#1}}

\newcommand{\eqIsgd}{\eqref{eq:implicit}\xspace}
\newcommand{\eqAisgd}{\eqref{eq:averaging}\xspace}
\newcommand{\remark}{
\noindent \emph{Remarks.}
}
\newcommand{\constLipZero}{\lambda_0}
\newcommand{\constLipOne}{\lambda_1}
\newcommand{\constLipTwo}{\lambda_2}

\newcommand{\constGamma}{\gamma}

\newcommand{\Fisher}[1]{\mathcal{I}(#1)}
\newcommand{\Fisherobs}[1]{\hat{\mathcal{I}}(#1)}
\newcommand{\minF}{\phi}
\newcommand{\maxF}{\minF}
\newcommand{\minEigF}{\underline{\lambda_f}}

\newcommand{\Fn}[1]{\mathcal{F}_{#1}}

\newcommand{\bias}[1]{||\thetaim{n}-\thetastar||^2}

\newcommand{\StateAssumptions}{
Let $\Fn{n} = \{\theta_0, \xi_1, \xi_2, \ldots, \xi_{n}\}$ denote the
filtration that process $\theta_n$ \eqIsgd\ is adapted to.
The norm $||\cdot||$ will denote the $L_2$ norm.
The symbol $\triangleq$ indicates a definition, and the symbol $\defeq$ denotes ``equal by definition''. For example, $x \triangleq y$ defines $x$ as equal to known variable $y$, 
whereas $x \defeq y$ denotes that $x$ is equal to $y$ by definition. We will not use 
this formalism when defining constants.
For two positive sequences $a_n, b_n$, we write
$b_n = \bigO{a_n}$ if there exists a fixed $c>0$ such that $b_n \le c a_n$, for all $n$;
also, $b_n = \littleO{a_n}$ if $b_n / a_n \to 0$.
When a positive scalar sequence $a_n$ is monotonically decreasing to zero,
we write $a_n \downarrow 0$.
Similarly, for a sequence $X_n$ of vectors or matrices,
$X_n = \bigO{a_n}$ denotes that $||X_n|| = \bigO{a_n}$,
and $X_n = \littleO{a_n}$ denotes that $||X_n|| = \littleO{a_n}$.
For two matrices $A, B$, $A \preceq B$ denotes that $B-A$ is nonnegative-definite;
$\mathrm{tr}(A)$ denotes the trace of $A$.

We now introduce the main assumptions
pertaining to the theory of this paper.
\begin{assumption}
\label{A:linear_loss}
The loss function $\loss(\theta, \xi)$ is almost-surely differentiable.
The random vector $\xi$ can be decomposed as $\xi = (x, y)$, $x\in\Reals{p}, y\in\Reals{d}$,
 such that
\begin{align}
\label{eq:loss}
\loss(\theta, \xi) = \loss(x^\intercal\theta, y).
\end{align}
\end{assumption}

\begin{assumption}
 \label{A:gamma_n}
	The learning rate sequence $\{\gamma_n\}$ is defined as
	$\gamma_n = \gamma_1 n^{-\constGamma}$,
	where $\gamma_1>0$ and $\gamma \in (1/2, 1]$.
\end{assumption}

\begin{assumption}[Lipschitz conditions]
\label{A:Lip}
For all $\theta_1, \theta_2\in\Theta$, a combination
of the following conditions is satisfied almost-surely:
\begin{enumerate}[(a)]
\item  \label{A:LipZero}
The loss function $L$ is Lipschitz-continuous with parameter $\constLipZero$, i.e.,
\begin{align}
| L(\theta_1, \xi) - L(\theta_2, \xi) |  \le \constLipZero ||\theta_1 - \theta_2||,\nn
\end{align}

\item
 \label{A:LipOne}
The map $\nabla L$ is Lipschitz-continuous with
parameter $\constLipOne$, i.e.,
\begin{align}
|| \nabla L(\theta_1, \xi) - \nabla L(\theta_2, \xi) ||  \le \constLipOne ||\theta_1 - \theta_2||,\nn
\end{align}
\item \label{A:LipTwo}
The map $\nabla^2 L$ is Lipschitz-continuous with parameter $\constLipTwo$, i.e.,
\begin{align}
|| \nabla^2 L(\theta_1, \xi) - \nabla^2 L(\theta_2, \xi) ||  \le \constLipTwo ||\theta_1 - \theta_2||.\nn
\end{align}
\end{enumerate}
\end{assumption}

\begin{assumption}
\label{A:Fisher}
The observed Fisher information matrix, $\Fisherobs{\theta}\triangleq\nabla^2 \loss(\theta, \xi)$,
has non-vanishing trace, i.e., there exists $\minF>0$ such that
$\mathrm{tr}(\Fisherobs{\theta}) \ge \minF$, almost-surely, for
all $\theta\in\Theta$.
The expected Fisher information matrix, $\Fisher{\theta}\triangleq \Ex{\Fisherobs{\theta}}$, has minimum eigenvalue $0 < \minEigF \le \minF$, for all $\theta\in\Theta$.
\end{assumption}

\begin{assumption}
\label{A:errors}
The zero-mean random variable $W_\theta \triangleq \nabla \loss(\theta, \xi) - \nabla \exloss(\theta)$ is
square-integrable, such that, for a fixed positive-definite $\Sigma$,
\begin{align}
\Ex{W_{\thetastar} W_{\thetastar}^\intercal} \preceq \Sigma.\nn
\end{align}
\end{assumption}
\vspace{-5px}
}

\newcommand{\refAssumption}[1]{Assumption \ref{A:#1}}
\newcommand{\aLinearLoss}{\ref{A:linear_loss}}
\newcommand{\assumeLinearLoss}{\refAssumption{linear_loss}}
\newcommand{\aGamma}{\ref{A:gamma_n}}
\newcommand{\assumeGamma}{\refAssumption{gamma_n}}

\newcommand{\aLipZero}{\ref{A:Lip}\eqref{A:LipZero}}
\newcommand{\assumeLipZero}{\refAssumption{Lip}\eqref{A:LipZero}}
\newcommand{\aLipOne}{\ref{A:Lip}\eqref{A:LipOne}}
\newcommand{\assumeLipOne}{\refAssumption{Lip}\eqref{A:LipOne}}
\newcommand{\aLipTwo}{\ref{A:Lip}\eqref{A:LipTwo}}
\newcommand{\assumeLipTwo}{\refAssumption{Lip}\eqref{A:LipTwo}}

\newcommand{\aConvex}{\ref{A:convexity}}
\newcommand{\assumeConvex}{\refAssumption{convexity}}

\newcommand{\aFisher}{\ref{A:Fisher}}
\newcommand{\assumeFisher}{\refAssumption{Fisher}}
\newcommand{\aErrors}{\ref{A:errors}}
\newcommand{\assumeErrors}{\refAssumption{errors}}

\newcommand{\DefinitionDerivative}{
\vspace{5px}
\begin{definition}
\label{definition:derivative}
Suppose that \assumeLinearLoss\ holds.
For observation $\xi=(x, y)$, the first derivative with respect to the natural parameter $x^\intercal\theta$ is
denoted by $L'(\theta, \xi)$, and is defined as
\begin{align}
\label{eq:definition_derivative}
L'(\theta, \xi) \triangleq \frac{\partial L(\theta, \xi)}{\partial (x^\intercal\theta)}\defeq
 \frac{\partial L(x^\intercal\theta, y)}{\partial (x^\intercal\theta)}.
\end{align}
Similarly, $L''(\xi, \theta)\triangleq \frac{\partial L'(\theta, \xi)}{\partial (x^\intercal\theta)}$.
\end{definition}
}

\newcommand{\LemmaImplicit}{
\begin{lemma}
\label{lemma:implicit_algo}
Suppose that \assumeLinearLoss\ holds,
and consider functions $L', L''$ from Definition \ref{definition:derivative}.
Then, almost-surely,
\begin{align}
\label{eq:lemma_scaling}
\nabla \loss(\theta_n, \xi_n) = s_n\nabla \loss(\theta_{n-1}, \xi_n);
\end{align}
the scalar  $s_n$ satisfies the fixed-point equation,
\begin{align}
\label{eq:lemma_fp}
s_n\kappa_{n-1} =  \loss'\left(\theta_{n-1} - s_n\gamma_n\kappa_{n-1} x_n, \xi_n \right),
\end{align}
where $\kappa_{n-1}\triangleq \loss'(\theta_{n-1}, \xi_n)$.
Moreover, if $L''(\theta, \xi) \ge0$ almost-surely for all $\theta \in \Theta$,
then
\begin{align}
s_n \in  \begin{cases} [\kappa_{n-1},0) &\mbox{if } \kappa_{n-1} < 0, \\ 
[0, \kappa_{n-1}] & \mbox{otherwise.} \end{cases}\nn
\end{align}

\end{lemma}
}

\newcommand{\TheoremMSE}{
\vspace{5px}
\begin{theorem}
\label{theorem:mse2}
Suppose that Assumptions \aLinearLoss, \aGamma, \aLipZero,
and \aFisher\ hold. Define $\delta_n \triangleq \Ex{||\theta_n-\thetastar||^2}$, and constants
$\Gamma^2 = 4\constLipZero^2\sum\gamma_i^2<\infty$,
$\epsilon = (1+\gamma_1 (\maxF-\minEigF))^{-1}$,
and $\lambda =1+\gamma_1 \minEigF \epsilon$.
Then, there exists constant $n_0>0$ such that, for all $n>0$,
\begin{align}
\delta_n \le
&  (8 \constLipZero^2 \gamma_1\lambda/\minEigF\epsilon) n^{-\gamma} + e^{-\log \lambda \cdot n^{1-\gamma}}[\delta_0 + \lambda^{n_0} \Gamma^2].\nn
 \end{align}
\end{theorem}
\vspace{-20px}
}

\newcommand{\cTwo}{\frac{2\gamma+1}{\minEigF^{1/2} n \gamma_n}}
\newcommand{\cFour}{\frac{\constLipTwo}{2n \minEigF^{1/2}}}
\newcommand{\Kone}{(2K_3 \gamma_1^2\lambda/\minEigF\epsilon)^{1/2}}
\newcommand{\Ktwo}{[\zeta_0 + \lambda^{n_{0,2}} \Delta^3]^{1/2}}
\newcommand{\Kthree}{(8 \constLipZero^2 \gamma_1\lambda/\minEigF\epsilon)^{1/2}}
\newcommand{\Kfour}{[\delta_0 + \lambda^{n_{0, 1}} \Gamma^2]^{1/2}}
\newcommand{\TheoremMSEfourth}{
\vspace{5px}
\begin{theorem}
\label{theorem:mse4}
Suppose that Assumptions \aLinearLoss, \aGamma, \aLipZero,
and \aFisher\ hold. For a constant $K_3>0$, define $\zeta_n \triangleq \Ex{||\theta_n-\thetastar||^2}$, and constants
$\Delta^3\triangleq K_3 \sum\gamma_i^3<\infty$,
$\epsilon \triangleq (1+\gamma_1 (\maxF-\minEigF))^{-1}$,
and $\lambda \triangleq 1+\gamma_1 \minEigF \epsilon$.
Then, there exists constant $n_0$ such that, for all $n>0$,
\begin{align}
\zeta_n \le
&  (2K_3 \gamma_1^2\lambda/\minEigF\epsilon) n^{-2\gamma} + e^{-\log \lambda \cdot n^{1-\gamma}}[\zeta_0 + \lambda^{n_0} \Delta^3].\nn
 \end{align}
\end{theorem}
\vspace{-20px}
}

\newcommand{\TheoremMSEaisgd}{
\begin{theorem}
\label{theorem:mse2_aisgd}
Consider the \aisgd\ procedure \ref{eq:aisgd_averaging}
and suppose that Assumptions \aLinearLoss, \aGamma, \aLipZero, \aLipTwo, \aFisher, and \aErrors\ hold. Then,
\begin{align}
\label{eq:theorem:mse2_aisgd}
(\Ex{||\thetaBar{n}-\thetastar||^2})^{1/2}
\le
& \frac{1}{\sqrt{n}} \left(\mathrm{trace}(\nabla^2 \exloss(\thetastar)^{-1} \Sigma \nabla^2 \exloss(\thetastar)^{-1})\right)^{1/2} \nn \\
& +  \frac{2\gamma+1}{\minEigF^{1/2} \gamma_1}
\Kthree n^{-1+\gamma/2} \nn \\
& + \cTwo \Kfour e^{-\log \lambda \cdot n^{1-\gamma}/2}  \nn \\
& + \frac{\constLipTwo}{2 \minEigF^{1/2}}
\Kone  n^{-\gamma} \nn \\
 & + \cFour \Ktwo K_2(n).
\end{align}
where  $K_2(n) = \sum_{i=1}^n\exp\left(- \log\lambda \cdot i^{1-\gamma}/2\right)$, and constants $\lambda, \epsilon, n_{0,1}, \delta_0, \Gamma^2$ are defined in Theorem \ref{theorem:mse2} (susbtituting $n_0$ for $n_{0,1}$),
and $\zeta_0, n_{0,2}, \Delta^3$ are defined in Theorem \ref{theorem:mse4}, 
substituting ($n_0$ for $n_{0,2}$).
\end{theorem}
}

\newcommand{\TheoremMSEaisgdSHORT}{
\vspace{5px}
\begin{theorem}
\label{theorem:mse2_aisgd_short}
Consider the \aisgd\ procedure \eqAisgd,
and suppose that Assumptions \aGamma, \aLipZero, \aLipTwo, \aFisher, and \aErrors\ hold. Then,
\begin{align}
(\Ex{||\thetaBar{n}-\thetastar||^2})^{1/2}
\le
& \frac{1}{\sqrt{n}} \left(\mathrm{tr}(\nabla^2 \exloss(\thetastar)^{-1} \Sigma \nabla^2 \exloss(\thetastar)^{-1})\right)^{1/2} \nn\\
 & + \bigO{n^{-1+\gamma/2}} + \bigO{n^{-\gamma}} \nn\\
 & + \bigO{\exp(-\log\lambda\cdot n^{1-\gamma}/2}.\nn
\end{align}
\vspace{-30px}
\end{theorem}
}

\bibpunct{(}{)}{;}{a}{,}{,}
\DeclareMathOperator*{\argmin}{arg\,min}
\begin{document}

\twocolumn[
\aistatstitle{Towards stability and optimality in stochastic gradient descent}
\aistatsauthor{ Panos Toulis \And Dustin Tran \And Edoardo M. Airoldi }
\aistatsaddress{ Harvard University \And Harvard University \And Harvard University }
]

\begin{abstract}
Iterative procedures for parameter estimation based on \gls{SGD} allow the estimation to scale to massive data sets.
However, in both theory and practice, they suffer from
numerical instability. Moreover, they are statistically inefficient
as estimators of the true parameter value.
To address these two issues, we propose a new iterative procedure termed
\emph{\gls{AISGD}}.
For statistical efficiency, \gls{AISGD}\ employs averaging of the iterates,
which achieves the optimal Cram\'{e}r-Rao bound under strong convexity, i.e., it is an optimal unbiased estimator of the true parameter value.
For numerical stability, \gls{AISGD} employs an implicit update
at each iteration, which is related to proximal operators in optimization.
In practice, \gls{AISGD} achieves competitive performance with other
state-of-the-art procedures. Furthermore,
it is more stable than averaging procedures that do not employ proximal updates, and is simple to implement as it requires fewer tunable hyperparameters than procedures that do employ proximal updates.
\end{abstract}

\section{Introduction}
\label{section:introduction}
The majority of problems in statistical estimation
can be cast as finding the parameter value $\thetastar \in \Theta$ such that
\begin{align}
\label{eq:problem}
\theta_\star = \arg \min_{\theta \in \Theta}  \Ex{\loss(\theta, \xi)},
\end{align}
where the expectation is with respect to
the random variable $\xi \in \Xi \subseteq \Reals{d}$
that represents the data, $\Theta \subseteq \Reals{p}$ is the parameter space, and $\loss : \Theta \times \Xi \to \Reals{}$ is a loss function.
A popular procedure for solving Eq.\eqref{eq:problem}
is \glsreset{SGD}\gls{SGD} \citep{zhang2004solving, bottou2004stochastic}, where
a sequence $\theta_n$ approximates $\thetastar$, and is updated iteratively,
one data point at a time, through the iteration
\begin{align}
\label{eq:sgd_explicit}
\theta_n = \theta_{n-1} - \gamma_n \nabla \loss(\theta_{n-1}, \xi_n),
\end{align}
where $\{\xi_1, \xi_2, \ldots\}$ is a stream of i.i.d. realizations of $\xi$, and $\{\gamma_n\}$ is a non-increasing sequence of positive real numbers, known as the learning rate.
The $n$th iterate $\theta_n$ in \gls{SGD} \eqref{eq:sgd_explicit}
can be viewed as an {\em estimator} of $\thetastar$.
To evaluate such iterative estimators it is typical to consider
three properties: convergence rate and numerical stability, by studying the
mean-squared errors $\Ex{||\theta_n-\thetastar||^2}$; and statistical
efficiency, by studying the limit $n
\Var{\theta_n}$, as $n\to\infty$.

While computationally efficient, the \gls{SGD} procedure \eqref{eq:sgd_explicit}
suffers from numerical instability and statistical inefficiency.
Regarding stability, \sgd\ is sensitive to specification of the learning rate $\gamma_n$, since the mean-squared errors can diverge arbitrarily when $\gamma_n$ is misspecified
with the respect to problem parameters, e.g., the convexity and Lipschitz parameters of the loss function \citep{benveniste1990adaptive, moulines2011non}.
Regarding statistical efficiency,
\gls{SGD} loses statistical information.
In fact, the amount of information loss depends on the misspecification of $\gamma_n$
with respect to the spectral gap of the matrix $\Ex{\nabla^2 L(\thetastar, \xi)}$ \citep{toulis2014statistical},
also known as the Fisher information matrix.
Several solutions have been proposed to resolve these two issues, e.g., using projections and gradient clipping. However, they are usually heuristic and hard to generalize.

In this paper, we aim for the ideal combination of computational efficiency,
numerical stability, and statistical efficiency using the following procedure:
\begin{minipage}{\columnwidth}
\begin{minipage}[!t]{0.2\columnwidth}
\hfill
\textbf{\gls{AISGD}}
\end{minipage}
\begin{minipage}[!t]{0.8\columnwidth}
\centering
\begin{flalign}
\label{eq:implicit}
\theta_n  = \theta_{n-1} - \gamma_n \nabla \loss(\theta_n, \xi_n)\hfill,\\
\label{eq:averaging}
\bar{\theta}_n = (1/n) \sum_{i=1}^n \theta_i.
\end{flalign}
\end{minipage}
\end{minipage}

Our proposed procedure, termed
\glsreset{AISGD}\emph{\gls{AISGD}}, is comprised of two inner procedures.
The first procedure employs updates given in Eq.\eqref{eq:implicit}, which are {\em implicit} because the iterate $\theta_n$ appears on both sides of the equation.
Procedure \eqref{eq:implicit}, also known as {\em implicit SGD} \citep{toulis2014statistical}, aims to stabilize the updates of the classic \sgd\ procedure \eqref{eq:sgd_explicit}.
In fact, implicit \sgd\ can be motivated
as the limit of a sequence of improved classic \sgd\ procedures.
To see this, first fix the sample history $\Fn{n-1} = \{\theta_0^s, \xi_1, \xi_2, \ldots, \xi_{n-1}\}$, where we use the superscript ``s'' in the classic \sgd\ procedure
in order to distinguish from implicit \sgd. Then,  $\theta_n^s = \theta_{n-1}^s - \gamma_n \nabla \loss(\theta_{n-1}^s, \xi_n) \triangleq \theta_{n}^{(1)}$.
If we ``trust'' $\theta_n^{(1)}$ to be a better estimate of $\thetastar$
than $\theta_{n-1}^s$, then we can use $\theta_n^{(1)}$ instead of $\theta_{n-1}^s$ in computing the loss function at data point $\xi_n$.
This leads to a revised update $\theta_n^s = \theta_{n-1}^s - \gamma_n \nabla L(\theta_n^{(1)}, \xi_n) \triangleq \theta_n^{(2)}$. Likewise, we can use $\theta_n^{(2)}$ instead of $\theta_n^{(1)}$, and so on.
If we repeat this argument ad infinitum, then we get
the following sequence of improved \sgd\ procedures,
\begin{align}
\label{eq:motivation}
\theta_n^s & = \theta_{n-1}^s - \gamma_n \nabla L(\theta_{n-1}, \xi_n),\nn\\
\theta_n^s & = \theta_{n-1}^s - \gamma_n \nabla L(\theta_n^{(1)}, \xi_n),\nn\\
\theta_n^s & = \theta_{n-1}^s - \gamma_n \nabla  L(\theta_n^{(2)}, \xi_n), \nn\\
& \ldots \nn\\
\theta_n^s & = \theta_{n-1}^s - \gamma_n \nabla L(\theta_n^{(\infty)}, \xi_n),
\end{align}
where $\theta_n^{(i)} = \theta_{n-1}^s - \nabla L(\theta_n^{(i-1)}, \xi_n)$,
with initial condition $\theta_n^{(0)} = \theta_{n-1}^s$.
In the limit, assuming a unique fixed point is reached almost surely,
the final procedure of sequence \eqref{eq:motivation} satisfies
$\theta_n^s  = \theta_{n-1}^s - \gamma_n \nabla L(\theta_n^{(\infty)}, \xi_n)
 = \theta_n^{(\infty)}$. This can be rewritten as $\theta_n^s = \theta_{n-1}^s - \gamma_n \nabla L(\theta_n^s, \xi_n)$, which is equivalent to implicit \sgd. Thus, implicit \sgd\ can be viewed
as a repeated application of classic \sgd, where we keep
updating the  same iterate $\theta_{n-1}^s$ using the same data point $\xi_n$,
until a fixed-point is reached. Nesterov's accelerated gradient, a
popular improvement of classic \gls{SGD}, is only one
application of this procedure.

The stability improvement achieved by implicit updates
can be motivated by the following argument.
Assume for simplicity that $L$ is strongly convex, almost surely,
with parameter $\mu>0$. Then for the implicit \sgd\ procedure \eqref{eq:implicit},
\begin{align}
 \theta_n + \gamma_n \nabla L(\theta_n, \xi_n) &= \theta_{n-1}, \nn\\
 ||\theta_n-\thetastar||^2 + 2\gamma_n (\theta_n-\thetastar)^\intercal&
\nabla L(\theta_n, \xi_n) \le ||\theta_{n-1}-\thetastar||^2,
\nn\\
 (1+\gamma_n \mu) ||\theta_n-\thetastar||^2  &\le ||\theta_{n-1}-\thetastar||^2,\nn\\
  ||\theta_n-\thetastar||^2  &\le \frac{1}{1+\gamma_n\mu} ||\theta_{n-1}-\thetastar||^2,\nn
\end{align}
which implies that $||\theta_n-\thetastar||^2$ is contracting almost
surely. In contrast, the classic \gls{SGD} procedure does not share
this contracting property.

While the implicit updates of Eq.\eqref{eq:implicit} aim to achieve
stability, the averaging of the iterates in Eq.\eqref{eq:averaging} aims to achieve statistical optimality.
\citet{ruppert1988efficient} gave a nice intuition on why iterate
averaging can lead to statistical optimality. When the learning rate
is $\gamma_n\propto n^{-1}$, then $\thetaBar{n}-\thetastar$ is a weighted average of $n$ error variables $\nabla L(\theta_{i-1}, \xi_i)$,
which therefore are significantly autocorrelated.
However, when $\gamma_n\propto n^{-\gamma}$ with $\gamma \in (0, 1)$, then $\thetaBar{n} - \thetastar$ is the average of $n^\gamma \log n$
error variables, which become uncorrelated in the limit. Thus, averaging improves the estimation accuracy.

\subsection{Related work}
\label{section:related}
The implicit update \eqIsgd\ is equivalent to
\begin{align}
\label{eq:proximal}
\theta_n = \argmin_{\theta \in \Theta}
\left\{ \frac{1}{2\gamma_n} ||\theta-\theta_{n-1}||^2 + \loss(\theta, \xi_n)\right  \}.
\end{align}
Arguably, the first method that used an update
similar to \eqref{eq:proximal} for estimation was the normalized least-mean squares filter of \citet{nagumo1967learning}, used in signal processing.
This update is also used by the {\em incremental proximal method} in optimization \citep{bertsekas2011incremental},
and has shown superior performance to classic \gls{SGD} both in theory and applications \citep{bertsekas2011incremental, toulis2014statistical, defossez2015averaged, toulis2015implicit}.
In particular, implicit updates lead to similar convergence rates as classic \gls{SGD} updates, but are significantly more stable.
This stability can also be motivated from a
Bayesian interpretation  of Eq.\eqref{eq:proximal}, where $\theta_n$ is
the posterior mode of a model with the standard multivariate normal $\mathcal{N}(\theta_{n-1}, \gamma_n I)$ as the prior,
$L(\theta, \cdot)$ as the log-likelihood,
and $\xi_n$ as the observation.

A statistical analysis of procedure \eqIsgd\ without averaging
was done by \citet{toulis2014statistical} who derived the asymptotic variance
$\Var{\theta_n}$ of $\theta_n$, and provided an algorithm to efficiently solve the fixed-point equation \eqIsgd\ for $\theta_n$ in the family of generalized linear models, which we generalize in this current work.
In the online learning literature, \citet{kivinen2006p} and
\citet{kulis2010implicit} have also analyzed implicit updates;
\citet{schuurmans2007implicit} have
further  applied implicit procedures on learning with kernels.
Notably the implicit update \eqref{eq:proximal} is related to the
importance weight updates proposed by \citet{karampatziakis2010online}, but the two
update forms are not equivalent, and are usually combined in practice \citep[Section 5]{karampatziakis2010online}.

Assuming that the expected loss $\exloss$ is known,
instead of update \eqref{eq:proximal} we could use the update
\begin{align}
\label{eq:proximal2}
\theta_n^+ = \arg \min_{\theta \in \Theta}
\left\{ \frac{1}{2\gamma_n} ||\theta-\theta_{n-1}||^2 + \exloss(\theta)\right  \}.
\end{align}

In optimization, this mapping from $\theta_{n-1}$ to $\theta_n^+$ in
Eq. \eqref{eq:proximal2} is known as a {\em proximal operator},
and is a special instance of the proximal point algorithm \citep{rockafellar1976monotone}.
Thus implicit \gls{SGD} involves mappings that are stochastic versions
of mappings from proximal operators.
The stochastic proximal gradient algorithm \citep{singer2009efficient, parikh2013proximal, rosasco2014convergence} is related but
different to implicit \gls{SGD}.
In contrast to implicit \gls{SGD}, the stochastic proximal gradient algorithm
first makes a classic \sgd\ update (forward step), and then an
implicit update (backward step). Only the forward step is stochastic
whereas the backward proximal step is not. This may increase convergence speed but may also introduce instability due to the forward step.

Interest on proximal operators has surged in recent years
because they are non-expansive and converge with minimal assumptions. Furthermore, they can be applied on non-smooth
objectives, and can easily be combined
in modular algorithms for optimization in large-scale and distributed
settings \citep{parikh2013proximal}.
The idea has also been generalized through splitting algorithms \citep{lions1979splitting, beck2009fast, singer2009efficient, duchi2011adaptive}.
\citet{krakowskigeometric} and
\citet{nemirovski2009robust} have shown that proximal methods
can fit better in the geometry of the parameter space $\Theta$, and
\citet{toulis2014implicit} have made a connection  to shrinkage methods in statistics.

Two recent procedures based on stochastic proximal updates are \proxsvrg\ \citep{xiao2014proximal}
and \proxsag\ \citep[Section 6]{schmidt2013minimizing}.
The main idea in both methods is to periodically compute an estimate of the
full gradient averaged over all data points in order to reduce the variance of
stochastic gradients. This requires a finite data setting, whereas \gls{AISGD}
also applies to streaming data.
Moreover, the periodic calculations in \proxsvrg\ are controlled by
additional hyperparameters, and the periodic calculations in
\proxsag\ require storage of the full gradient at every iteration.
\gls{AISGD} differs because it employs averaging to achieve statistical
efficiency, has no additional hyperparameters or major storage requirements,
and thus it has a simpler implementation.

Averaging of the iterates in Eq.\eqAisgd\ is the other key component of \gls{AISGD}. Averaging was proposed
and analyzed in the stochastic approximation literature by \citet{ruppert1988efficient} and
\citet{bather1989stochastic}. \citet{polyak1992acceleration} substantially
expanded the scope of the averaging method by proving asymptotic optimality of the classic \gls{SGD} procedure with averaging, under suitable assumptions.
Their results showed clearly that slowly-convergent stochastic approximations  (achieved
when the learning rates are large)
need to be averaged.
Recent work has analyzed classic \sgd\ with averaging
\citep{zhang2004solving, xu2011towards, shamir2012stochastic, bach2013non} and has shown their superiority in numerous learning tasks.

\subsection{Overview of results}
\label{section:overview}
In this paper, we study the iterates $\theta_n$ and use the results
to study $\thetaBar{n}$ as an estimator of $\thetastar$.
Under strong convexity of the expected loss,
we derive upper bounds for the squared errors
$\Ex{||\theta_n-\thetastar||^2}$ and $\Ex{\thetaBar{n}-\thetastar||^2}$
in Theorem \ref{theorem:mse2} and Theorem \ref{theorem:mse2_aisgd_short},
respectively. In the supplementary material,
we also give bounds for $\Ex{||\theta_n-\thetastar||^4}$.

Two main results are derived from our theoretical analysis.
First, $\thetaBar{n}$ achieves the Cram\'{e}r-Rao bound, i.e., no other unbiased estimator of $\theta_\star$ can do better in the limit, which is equivalent to the optimal $\bigO{1/n}$ rate of convergence for first-order procedures.
Second, \gls{AISGD} is significantly more stable to
misspecification of the learning rate relative to classic averaged \gls{SGD} procedures, with
respect to the learning problem parameters, e.g.,  convexity and Lipschitz constants.
Finally, we perform experiments on several standard machine learning
tasks, which show that \gls{AISGD} comes closer to combining
stability, optimality, and simplicity than other competing methods.

\section{Preliminaries}
\label{section:preliminaries}

\parhead{Notation.}
\StateAssumptions
\remark\
\assumeLinearLoss\ puts a constraint on the loss function,
but it is not very restrictive because the majority of machine
learning models indeed depend on the parameter $\theta$ through a linear
combination with features.
A notable exception includes loss functions with a regularization term.
Although it is easy to add regularization to \aisgd\ we will not do so in this paper because \gls{AISGD} works well without it, since
the proximal operator \eqref{eq:proximal} already regularizes the estimate $\theta_n$ towards $\theta_{n-1}$. In experiments, regularization neither improved nor worsened \gls{AISGD} (see
supplementary material for more details).
\assumeGamma\ on learning rates and \assumeErrors\ are standard in the literature of stochastic approximations, dating
back to the original paper of \citet{robbins1951stochastic} in the
one-dimensional parameter case.

Assumptions on Lipschitz gradients (\assumeLipOne, \assumeLipTwo) can be relaxed; for example, \citet{benveniste1990adaptive}
relax this assumption using  $||\theta_1-\theta_2||^q$.
However, these two Lipschitz conditions are commonly used
in order to simplify the non-asymptotic analysis \citep{moulines2011non}.
\assumeLipZero\ is less standard in classic \sgd\ literature
but has so far been standard in the limited literature on implicit
\sgd\ \citep{bertsekas2011incremental}. We can forgo this assumption and still maintain identical rates for
the errors, although at the expense of a more complicated analysis.
It is also an open problem whether a nice stability result
similar to Theorem \ref{theorem:mse2} can be derived under
\assumeLipOne\ instead of \assumeLipZero.
We discuss this issue after the proof of Theorem \ref{theorem:mse2} in the supplementary material.

\assumeFisher\ makes two claims. The first claim on the observed Fisher information matrix
is a relaxed form of strong convexity for the loss $L(\theta, \xi)$. However, in contrast
to strong convexity, this claim allows several eigenvalues of $\nabla^2 L$ to be zero.
The second claim of \assumeFisher\ is equivalent to strong convexity of the
expected loss $\exloss(\theta)$.
From a statistical perspective, strong convexity posits that there is
information in the data for all elements of $\thetastar$.
This assumption is necessary
to derive bounds on the errors $\Ex{||\theta_n-\thetastar||^2}$, and
has been used to show optimality of classic
\sgd\ with averaging \citep{polyak1992acceleration, ljung1992stochastic,
xu2011towards, moulines2011non}.

Overall, our assumptions are weaker than the assumptions
in the limited literature on implicit \sgd. For example,
\citet[Assumptions 3.1, 3.2]{bertsekas2011incremental} assumes
almost-sure bounded gradients $\nabla L(\theta, \xi)$ in addition to
\assumeLipZero; \citet{ryustochastic} assume strong convexity of $L(\theta, \xi)$, in expectation, which can simplify the analysis significantly. We discuss more details in the supplementary material after the proof of Theorem \ref{theorem:mse2}.

\section{Theory}
\label{section:theory}

In this section we present our theoretical
analysis of \gls{AISGD}.
All proofs are given in the supplementary
material.
The main technical challenge in analyzing implicit \gls{SGD} \eqIsgd\
is that
unlike typical analysis with classic \gls{SGD} \eqref{eq:sgd_explicit},
the error $\xi_n$ is not conditionally independent of $\theta_n$.
This implies that $\ExCond{\nabla \loss(\theta_n, \xi_n)}{\theta_n} \neq
\exloss(\theta_n)$,
which makes it no longer possible to use the convexity properties of $\exloss$ to
analyze the errors $\Ex{||\theta_n-\thetastar||^2}$, as it is common in the
literature.

As mentioned earlier, to circumvent this issue other authors have made strict almost-sure assumptions on the implicit procedure \eqIsgd \citep{bertsekas2011incremental, ryustochastic}.
In this paper, we rely on weaker conditions, namely the
Lipschitz assumptions \aLipZero-\aLipTwo,
which are also used in non-implicit procedures.
Our proof strategy relies on a master lemma (Lemma 3 in supplementary
material) for the analysis of recursions that appear to be typical in implicit
procedures. This result is novel to our best knowledge, and it can be useful
in future research on implicit procedures.

\subsection{Computational efficiency}
Our first result enables efficient computation of the implicit update \eqIsgd.
In general, this can be expensive due to
solving a fixed-point equation in many dimensions, at every iteration.
We reduce this multi-dimensional equation to an equation
of only one dimension.
Furthermore, under almost-sure convexity of the loss function,
efficient search bounds for the one-dimensional
fixed-point equation are available. This result generalizes an earlier result in
efficient computation of implicit updates on generalized linear models \citep[Algorithm 1]{toulis2014statistical}.

\DefinitionDerivative\
\LemmaImplicit

\remark\
Lemma \ref{lemma:implicit_algo} has two parts.
First, it shows that the implicit update can be performed by obtaining $s_n$
from the fixed-point Eq.\eqref{eq:lemma_scaling}, and then using $
\nabla \loss(\theta_n, \xi_n) = s_n \nabla \loss(\theta_{n-1}, \xi_n)$
in the implicit update \eqIsgd. The fixed-point equation
can be solved through a numerical root-finding procedure \citep{kivinen2006p, kulis2010implicit, toulis2014statistical}.
Second, when the loss function is convex, then
narrow search bounds for $s_n$ are available.
This property holds, for example, when the loss function is the
negative log-likelihood in an exponential family.

\subsection{Non-asymptotic analysis}
Our next result is on the mean-squared errors $\Ex{||\theta_n-\thetastar||^2}$. These errors show the stability and convergence rates of implicit \sgd\, and are used in combination with bounds
on errors $\Ex{||\theta_n-\thetastar||^4}$ to
derive bounds on the errors $\Ex{||\thetaBar{n}-\thetastar||^2}$
of the averaged procedure.\footnote{
The bounds for the fourth moments
$\Ex{||\theta_n-\thetastar||^4}$ are given in the supplementary material because they rely on the same intermediate results as $\Ex{||\theta_n-\thetastar||^2}$.}

\TheoremMSE\

\remark\
According to Theorem \ref{theorem:mse2},
the convergence rate of the implicit iterates $\theta_n$
is $\bigO{n^{-\gamma}}$. This matches
earlier results on rates of classic \sgd\ \citep{benveniste1990adaptive, moulines2011non}. The most important difference, however,
is that the implicit procedure discounts the initial conditions
$\delta_0$ at an exponential rate, regardless of the specification
of the learning rate. As shown by \citet[Theorem 1]{moulines2011non}, in
classic \gls{SGD} there exists
a term $\exp(\constLipOne^2 \gamma_1^2 n^{1-2\gamma})$
in front of the initial conditions, which can be catastrophic
if the learning rate parameter $\gamma_1$ is misspecified.
In contrast, the implicit iterates are unconditionally stable,
i.e., any specification of the learning rate will lead to a
stable discounting of the initial conditions.

\TheoremMSEaisgdSHORT\

\remark\
The full version of Theorem \ref{theorem:mse2_aisgd_short}, which
includes all constants, is given in the supplementary material.
Even in its shortened form, Theorem \ref{theorem:mse2_aisgd_short} delivers three main results.
First, the iterates $\thetaBar{n}$ attain the Cram\'{e}r-Rao lower bound, i.e., any other unbiased estimator of $\theta_\star$ cannot have lower
MSE than $\thetaBar{n}$. From an optimization
perspective, $\thetaBar{n}$ attains the rate $\bigO{1/n}$,
which is optimal for first-order methods \citep{nesterov2004introductory}.
This result matches the asymptotic optimality of averaged iterates from
classic \gls{SGD} procedures, which has been proven by \citet{polyak1992acceleration}.

Second, the remaining rates are $\bigO{n^{-2+\gamma}}$
and $\bigO{n^{-2\gamma}}$. This implies the optimal
choice $\gamma=2/3$ for the exponent of the learning rate.
It extends the results of \citet{ruppert1988efficient}, and more recently by
\citet{xu2011towards}, and \citet{moulines2011non}, on optimal exponents
for classic \gls{SGD} procedures.

Third, as with non-averaged implicit iterates in Theorem \ref{theorem:mse2}, the averaged iterates $\thetaBar{n}$ have a decay of the initial conditions regardless
of the specification of the learning rate parameter.
This stability property is inherited from the underlying implicit \gls{SGD} procedure \eqIsgd\
that is being averaged. In contrast, averaged iterates of classic \gls{SGD} procedures can diverge numerically because arbitrarily large terms can appear in front of  initial conditions \citep[Theorem 3]{moulines2011non}.

\section{Experiments}
\label{section:experiments}
In this section, we show that \aisgd\ achieves comparable,
and sometimes superior, results to other methods while combining statistical efficiency, stability,
and simplicity. In our experiments, we compare our procedure to the following procedures:

\begin{itemize}
\item \gls{SGD}: Classic stochastic gradient descent in its standard formulation
\citep{sakrison1965efficient, zhang2004solving}, which employs the
update $\theta_n =
\theta_{n-1}- \gamma_n \nabla \loss(\theta_{n-1}, \xi_n)$.
\item \implicit: Stochastic gradient descent procedure introduced in
\citet{toulis2014statistical} which employs implicit update \eqIsgd\ without
averaging. It is robust to misspecification of the learning rate but also exhibits slower convergence
in practice relative to classic \sgd.
\item \asgd: Averaged stochastic gradient descent procedure with
classic updates of
the iterates \citep{xu2011towards, shamir2012stochastic, bach2013non}. This is equivalent to \gls{AISGD} where the update \eqIsgd\
is replaced by the classic step $\theta_n = \theta_{n-1}- \gamma_n
\nabla\loss(\theta_{n-1}, \xi_n)$.
\item \proxsvrg: A proximal version of the stochastic gradient descent procedure with progressive variance reduction (SVRG) \citep{xiao2014proximal}.
\item \proxsag:  A proximal version of the stochastic average gradient (SAG)
procedure \citep{schmidt2013minimizing}.  While its theory has not been formally
established, \proxsag\ has shown similar convergence properties to
\proxsvrg in practice.
\item \adagrad: A stochastic gradient descent procedure with a form of diagonal scaling to adapt the learning rate \citep{duchi2011adaptive}.
\end{itemize}

Note that \proxsvrg\ and \proxsag\ are applicable only to fixed data sets
and not to the streaming setting.
Therefore the theoretical linear convergence rate of these methods
refers to convergence to an empirical minimizer (e.g., maximum
likelihood, or maximum a-posteriori if there is regularization), and not to the ground truth $\theta_\star$.
On the other hand, \aisgd can be applied to both data settings.

We also note that \adagrad, and similar adaptive schedules,
\citep{tieleman2012lecture,kingma2015adam} effectively approximate
the natural gradient $\Fisher{\theta}^{-1} \nabla L(\theta, \xi)$ by using a
multi-dimensional learning rate. These learning rates have the added advantage
of being less sensitive than one-dimensional rates to tuning of
hyperparameters; they can be combined in practice with \gls{AISGD}.

\subsection{Statistical efficiency and stability}
\label{section:statistical}
We first demonstrate the theoretical results on the stability and statistical
optimality of \aisgd. To do so, we follow a simple normal linear regression example from \citet{bach2013non}.
Let $N=10^6$ be the number of observations, and $p=20$
be the number of features. Let $\thetastar=(0,0,\ldots,0)^\intercal$ be the ground truth. The random variable $\xi$ is decomposed
as $\xi_n = (x_n, y_n)$, where the feature vectors $x_1,\ldots,x_N \sim
\mathcal{N}_p(0, H)$ are i.i.d. normal random variables,
and $H$ is a randomly generated symmetric matrix with eigenvalues $1/k$, for $k=1,\ldots, p$.
The outcome $y_n$ is sampled from a normal distribution as $y_n\:\vert\: x_n\sim \mathcal{N}(x_n^\intercal\theta_*, 1)$, for $n=1,\ldots,N$. Our loss function is defined as the squared residual, i.e.,
$\loss(\theta, \xi_n) = (y_n - x_n^\intercal \theta)^2$,
and thus $\exloss(\theta) = \Ex{\loss(\theta, \xi)} = (\theta - \theta_\star)^\intercal H (\theta - \theta_\star)$.

We choose a constant learning rate $\gamma_n \equiv \gamma$ according to the average radius of the data $R^2=\mathrm{trace}(H)$, and for both \asgd\ and \gls{AISGD} we collect iterates $\theta_n$, $n=1,\ldots,N$, and keep the average $\thetaBar{n}$.
In Figure \ref{figure:normal}, we plot $\exloss(\thetaBar{n})$ for each iteration for a maximum of $N$ iterations in log-log space.
\begin{figure}[!htb]
\centering
\includegraphics[width=\columnwidth]{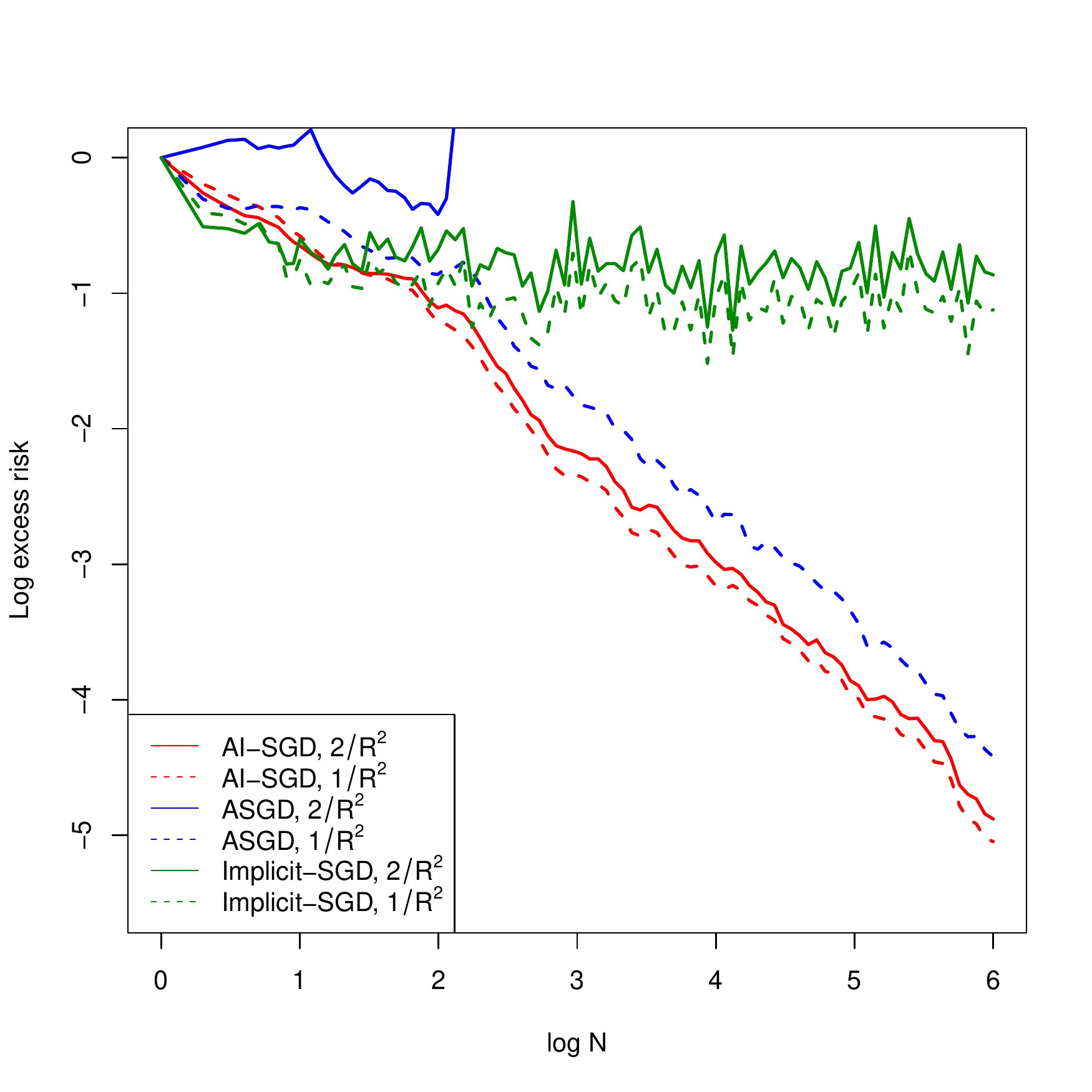}
\caption{Loss of \aisgd, \asgd, and \implicit, on simulated multivariate normal
data with $N=10^6$ observations, $d=20$ features.
The plot shows that \gls{AISGD} achieves
stability regardless of the specification of the learning rate
$\gamma_n\equiv\gamma$.
In contrast, \asgd\ diverges when the learning rate is only slightly misspecified (e.g., solid, blue line).
}
\label{figure:normal}
\end{figure}

Figure \ref{figure:normal} shows that \gls{AISGD} performs on par with \asgd\ for
the rates at which \asgd\ is known to be optimal.  However, the benefit of the
implicit procedure \eqIsgd\ in \gls{AISGD} becomes clear as the learning rate increases.
Notably, \gls{AISGD} remains stable for learning rates that are above the
theoretical threshold, i.e., when $\gamma > 1/R^2$, whereas \asgd\ diverges above that
threshold, e.g., when $\gamma=2/R^2$. This stable behavior is also exhibited in \implicit, but
\implicit\ converges at a slower rate than \gls{AISGD}, and thus does
not combine stability with statistical efficiency. This behavior is also
reflected for \gls{AISGD} when using decaying learning rates, e.g., $\gamma_n
\propto 1/n$.

\subsection{Classification error}
\label{section:classification}
We now conduct a study of \aisgd's empirical performance on standard benchmarks
of large-scale linear classification. For brevity, we display results on four
data sets, although we have seen similar results on eight additional ones (see
the supplementary material for more details).

Table \ref{datasets} displays a summary of the data sets.
The {COVTYPE} data
set \citep{blackard1998comparison}
consists of forest cover types in which the task is to classify class 2 among 7
forest cover types.
{DELTA} is synthetic data offered in the PASCAL Large Scale Challenge
\citep{sonnenburg2008pascal} and we apply the default processing offered by the
challenge organizers.
The task in {RCV1} is to classify documents belonging to class {CCAT} in the
text dataset \citep{lewis2004rcv1}, where we apply the standard preprocessing
provided by \citet{bottou2012stochastic}.  In the MNIST data set
\citep{lecun1998gradient} of images of handwritten digits, the task is to
classify digit 9 against all others.

For \gls{AISGD} and \asgd, we use the learning
rate $\gamma_n = \eta_0(1+\eta_0 n)^{-3/4}$ prescribed in
\citet{xu2011towards}, where the constant $\eta_0$ is determined through
preprocessing on a small
subset of the data.  Hyperparameters for other methods are set based on a
computationally intensive grid search over the entire hyperparameter space:
this includes step sizes for \proxsag, \proxsvrg, and \adagrad,
and the inner iteration count for \proxsvrg.
For all methods we use $L_2$ regularization with parameter $\lambda$ which
varies for each data set, and which is also used in \citet{xu2011towards}.

\begin{figure}[t!]
\hspace{-2em}
\includegraphics[width=1.1\columnwidth]{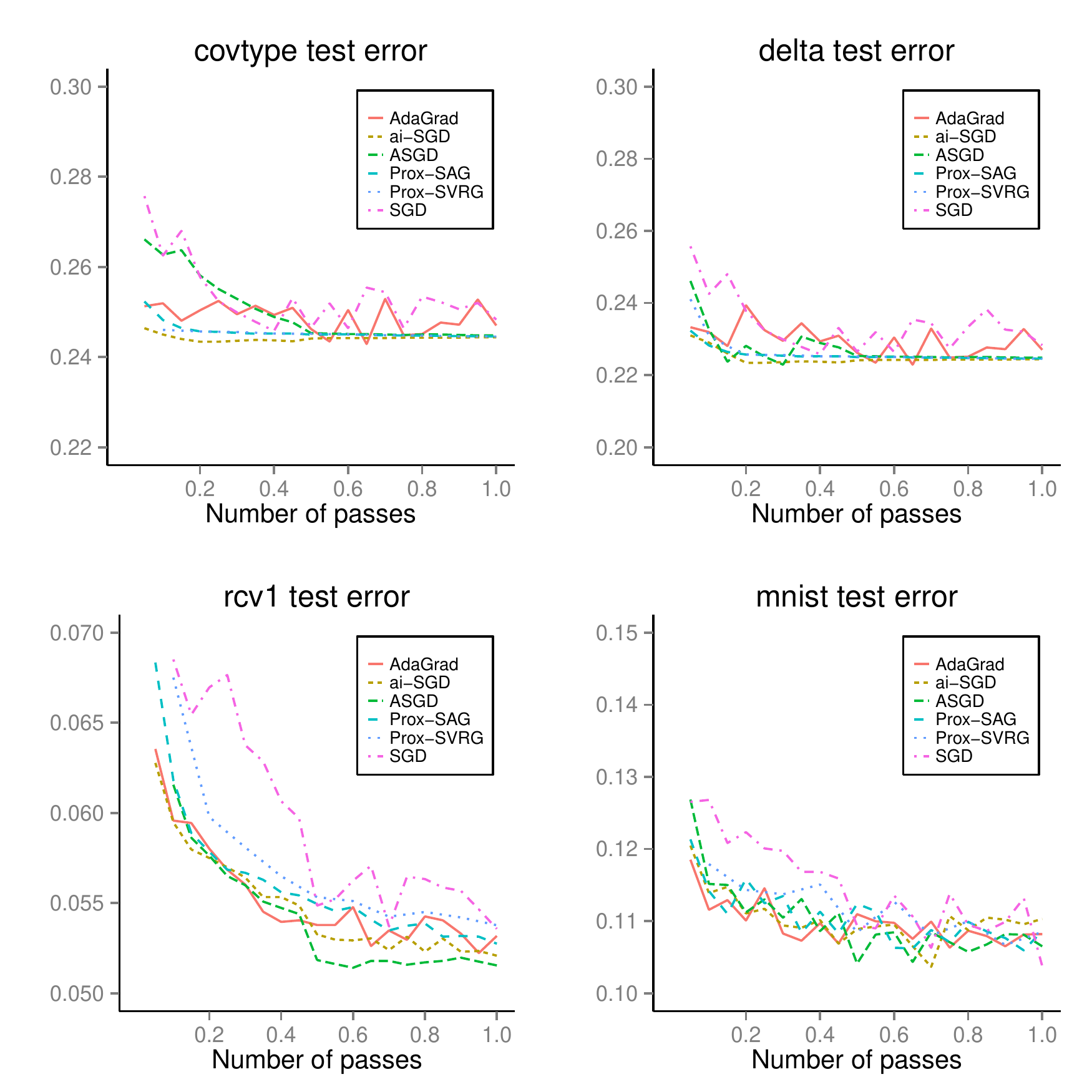}
\caption{Large scale linear classification with log loss on four data sets. Each
plot indicates the test error of various stochastic gradient methods over a single
pass of the data.}
\label{figure:classification}
\end{figure}
\begin{table*}[!htb]
  \centering
\begin{tabular}{|l|l|l|l|l|l|l|l|l|l|l|}
\hline
       & description          & type   & features& training set & test set & $\lambda$\\
\hline
covtype& forest cover type    & sparse & 54      & 464,809 & 116,203 & $10^{-6}$\\
delta  & synthetic data       & dense  & 500     & 450,000 & 50,000     & $10^{-2}$\\
rcv1   & text data            & sparse & 47,152  & 781,265 & 23,149  & $10^{-5}$\\
mnist  & digit image features & dense  & 784     & 60,000  & 10,000  & $10^{-3}$\\
\hline
\end{tabular}
\caption{\label{datasets} Summary of data sets and the $L_2$ regularization
parameter, following the settings in \citet{xu2011towards}.
}
\end{table*}

The results are shown in Figure \ref{figure:classification}.
We see that \gls{AISGD} achieves comparable performance with the tuned proximal
methods \proxsvrg\ and \proxsag, as well as \adagrad.
All methods have a comparable convergence
rate and take roughly a single pass in order to converge.
Interestingly,
\adagrad\ exhibits a larger variance in its estimate than the proximal methods. This comes from the less known fact that the learning rate in \adagrad\ is a suboptimal approximation of the Fisher information, and hence it is statistically inefficient.

\subsection{Sensitivity analysis}
\label{section:sensitivity}
We examine the inherent stability of the aforementioned procedures by perturbing their hyperparameters. That is, we perform
sensitivity analysis by varying any hyperparameters that the user must tweak in order to fine tune the convergence of each procedure.
We do so for hyperparameters in \asgd\ (the learning rate), \proxsvrg\ (proximal step size
$\eta$ and inner iteration $m$), and \aisgd (the learning rate).

\begin{figure}[!htb]
\centering
\includegraphics[width=0.7\columnwidth]{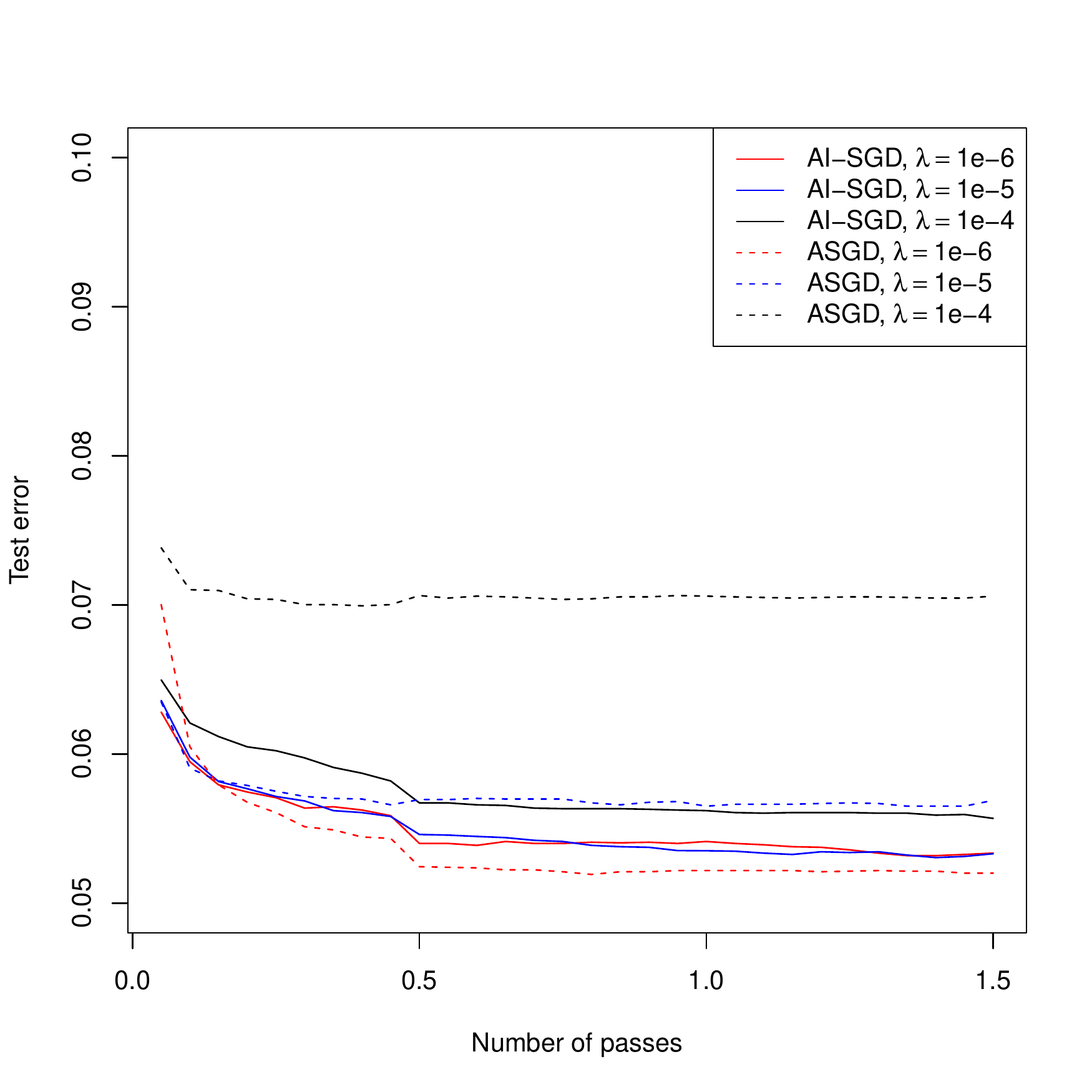}
\\[-4ex]
\includegraphics[width=0.7\columnwidth]{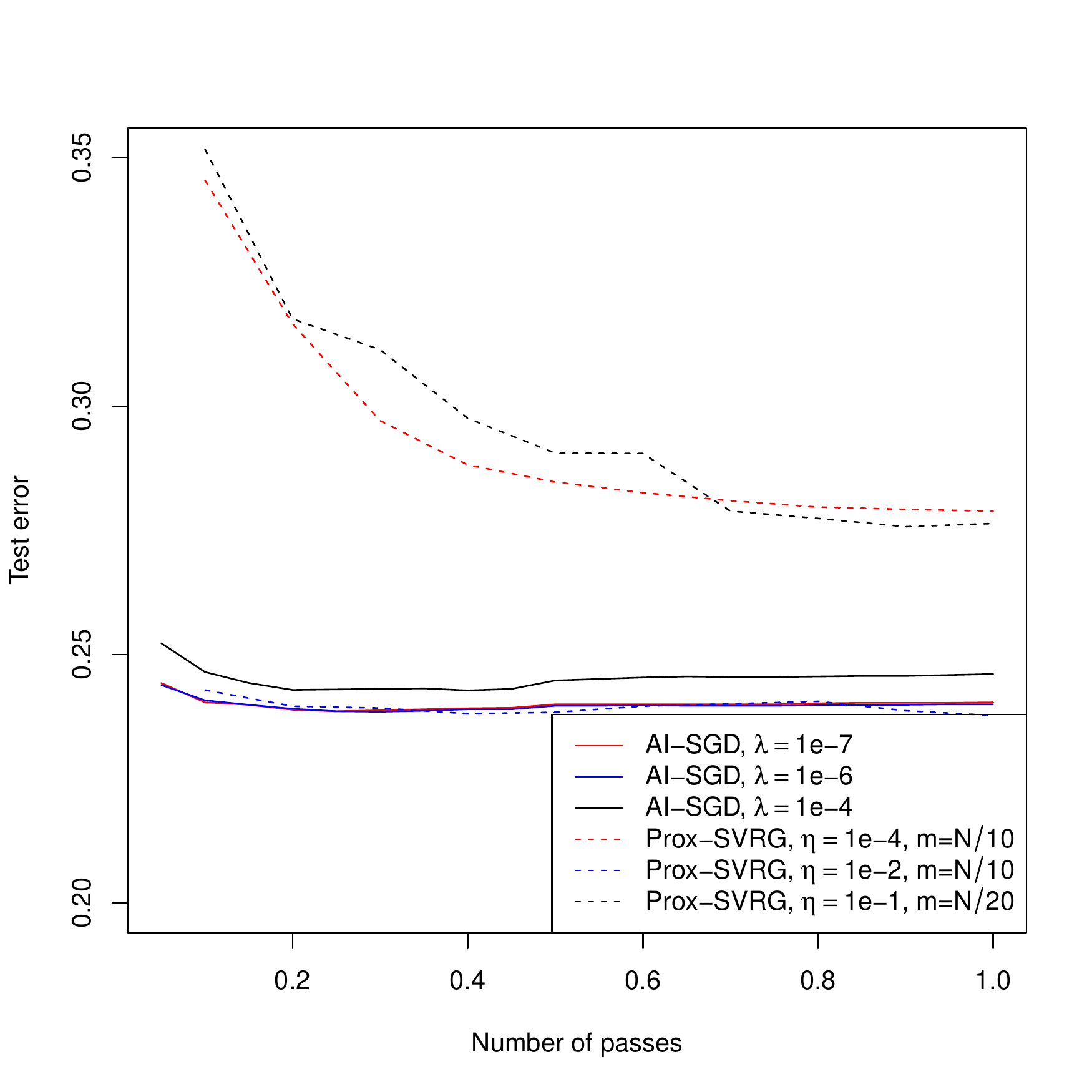}
\caption{Top: Logistic regression on the $\mathrm{RCV1}$ dataset,
performing sensitivity analysis of \aisgd\ and \asgd\ for the choice of
regularization parameter $\lambda$. Bottom: linear SVM on the
$\mathrm{covtype}$ dataset, performing sensitivity analysis of \aisgd\ and
\proxsvrg, in which \proxsvrg\ has additional hyperparameters $\eta$ according
to the step size of the proximal update and $m$ according to the inner
iteration count.}
\label{figure:sensitivity}
\end{figure}

The results are shown in Figure \ref{figure:sensitivity}.
When we decrease the regularization parameter, \asgd\ performs increasingly worse.  While it may converge, the test
error can be arbitrarily large. On the other hand, \aisgd\ always achieves convergence and is not affected by the choice of the hyperparameter.
When the regularization parameter is about $1/N$, e.g., when $\lambda <
1\eminus6$, \asgd\ remains stable and achieves the same performance as \aisgd.
Similar results hold when perturbing the hyperparameters $\eta$ and $m$ in
\proxsvrg, as \aisgd\ does not require specification of such hyperparameters.

\section{Conclusion}
\label{section:conclusion}
We propose a statistical learning procedure, termed \aisgd, and investigate its theoretical and empirical properties.
\gls{AISGD} combines simple stochastic proximal steps, also known as implicit updates, with iterate averaging and larger step-sizes.
The proximal steps allow \gls{AISGD} to be significantly more stable compared to classic \gls{SGD} procedures, with or without averaging of the iterates; this stability comes at virtually no computational cost for a large family of
machine learning models.
Furthermore, the averaging of the iterates lead \gls{AISGD} to be statistically optimal, i.e., the variance of the iterate $\thetaBar{n}$ of \aisgd\  achieves the minimum Cram\'{e}r-Rao lower bound, under strong convexity.
Last but not least, \gls{AISGD} is as simple to implement as classic
\gls{SGD}.  In comparison, other stochastic proximal procedures, such as
\proxsvrg\ or \proxsag, require tuning of
hyperparameters that control periodic calculations over the entire
dataset, and possibly storage of the full gradient.

\section*{References}
\renewcommand{\bibsection}{}
\bibliographystyle{apalike}
\bibliography{aistats2016}

\clearpage
\appendix
\onecolumn
\section{Note}
Lemmas 1, 2, 3 and 4, and Corollary 1, were originally derived by \citet{toulis2014implicit}. These intermediate results (and Theorem 1) provide the necessary foundation to derive Lemma 5 (only in this supplement) and Theorem 2 on the asymptotic optimality of $\thetaBar{n}$, which is the key result of the main paper.
We fully state these intermediate results here for convenience but we point the reader to the aforementioned reference for the proofs and for more details on the theory of (non-averaged) implicit stochastic gradient descent (implicit SGD).

\section{Introduction}
\label{section:introduction}
Consider a random variable $\xi \in \Xi$,
a parameter space $\Theta$
that is convex and compact,
and a loss function $\loss : \Theta \times \Xi \to \Reals{}$.
We wish to solve the following stochastic optimization problem:
\begin{align}
\label{eq:problem}
\theta_\star = \arg \min_{\theta \in \Theta}  \Ex{\loss(\theta, \xi)},
\end{align}
where the expectation is with respect to $\xi$.
Define the expected loss,
\begin{align}
\label{eq:exloss}
\exloss(\theta)  = \Ex{\loss(\theta, \xi)},
\end{align}
where $\loss$ is differentiable almost-surely.
In this work we study a stochastic approximation procedure to solve \eqref{eq:problem} defined through the iterations
\begin{align}
\label{eq:aisgd_implicit}
& \boldsymbol{\theta_n} = \theta_{n-1} -\gamma_n \nabla \loss(\boldsymbol{\theta_n}, \xi_n),
	\hspace{5px} \theta_0 \in \Theta, \\
\label{eq:aisgd_averaging}
& \bar{\theta}_n = \frac{1}{n} \sum_{i=1}^n \theta_i,
\end{align}
where $\{\xi_1, \xi_2, \ldots\}$ are i.i.d. realizations of $\xi$,
and $\nabla \loss(\theta, \xi_n)$ is the gradient of the loss function with respect to $\theta$ given realized value $\xi_n$.
The sequence $\{\gamma_n\}$ is a non-increasing sequence of positive real numbers.
We will refer to procedure
defined by \eqref{eq:aisgd_implicit} and \eqref{eq:aisgd_averaging} as \emph{averaged implicit stochastic gradient descent}, or \aisgd\ for short.
Procedure \aisgd\ combines two ideas, namely an implicit update in Eq. \eqref{eq:aisgd_implicit} as $\theta_n$ appears on
both sides of the update, and averaging of the iterates $\theta_n$ in Eq. \eqref{eq:aisgd_averaging}.

\renewcommand{\eqIsgd}{\eqref{eq:aisgd_implicit}}
\renewcommand{\eqAisgd}{\eqref{eq:aisgd_averaging}}

\section{Notation and assumptions}
\StateAssumptions

\section{Proof of Lemma \ref{lemma:implicit_algo}}
\DefinitionDerivative\
\LemmaImplicit\
\begin{proof}
See \citet[Theorem 4.1]{toulis2014implicit}.
\end{proof}
\section{Proof of Theorem \ref{theorem:mse2}}
\subsection{Useful lemmas}
In this section, we will present the intermediate lemmas on recursions that will be useful for the non-asymptotic analysis of the implicit procedures.

\begin{lemma}
\label{lemma:decay_factor}
Consider a sequence $b_n$ such that $b_n \downarrow 0$
and $\sum_{i=1}^\infty b_i = \infty$. Then,
there exists a positive constant $K>0$, such that
\begin{align}
\label{eq:decay_factor}
\prod_{i=1}^n \frac{1}{1+b_i} \le \exp(-K  \sum_{i=1}^n b_i).
\end{align}
\end{lemma}
\begin{proof}
See \citet[Lemma B.1]{toulis2014implicit}.
\end{proof}

\begin{lemma}
\label{lemma:implicit_recursion}
Consider scalar sequences $a_n \downarrow 0, b_n \downarrow 0$, and $c_n \downarrow 0$ such that, $a_n = \littleO{b_n}$,  and $A\triangleq \sum_{i=1}^\infty a_i < \infty$.
Suppose there exists $n'$ such that $c_n/b_n < 1$ for all $n>n'$.
Define,
\begin{align}
\label{new:defs}
\delta_n \triangleq \frac{1}{a_n} (a_{n-1}/b_{n-1} -a_n/{b_n}) \text{ and }
\zeta_n \triangleq \frac{c_n}{b_{n-1}} \frac{a_{n-1}}{a_n},
\end{align}
and suppose that $\delta_n \downarrow 0$ and $\zeta_n  \downarrow 0$.
Fix $n_0 > 0$ such that $\delta_n + \zeta_n < 1$
and $(1+c_n)/(1+b_n) < 1$, for
all $n \ge n_0$.

Consider a positive sequence $y_n>0$
that satisfies the recursive inequality,
\begin{align}
\label{ineq:implicit_recursion}
y_n \le \frac{1+c_n}{1 + b_n} y_{n-1}  + a_n.
\end{align}
Then, for every $n>0$,
\begin{align}
\label{eq:result:implicit_recursion}
y_n \le K_0 \frac{a_n}{b_n} + Q_{1}^n y_0 + Q_{n_0+1}^n (1+c_1)^{n_0} A,
 \end{align}
 where $K_0 = (1+b_1) \left(
 1-\delta_{n_0}  - \zeta_{n_0}\right)^{-1}$, and
 $Q_i^n = \prod_{j=i}^n (1+c_i)/(1+b_i)$,
such that $Q_i^n = 1$ if $n < i$, by definition.
\end{lemma}
\begin{proof}
See \citet[Lemma B.2]{toulis2014implicit}.
\end{proof}

\begin{corollary}
\label{corollary:implicit_recursion}
In Lemma \ref{lemma:implicit_recursion}
assume $a_n = a_1 n^{-\alpha}$ and $b_n = b_1 n^{-\beta}$, and $c_n=0$,
where $a_1, b_1, \beta>0$ and $\max\{\beta, 1\} < \alpha < 1+\beta$,
and $\beta \ne 1$.
Then,
\begin{align}
\label{thm1:eq4}
y_n \le 2\frac{a_1 (1+b_1)}{b_1} n^{-\alpha + \beta} + \exp(-\log(1+b_1)n^{1-\beta}) [y_0 +(1+b_1)^{n_0}A],
 \end{align}
 where $n_0>0$ and $A=\sum_i a_i < \infty$. If $\beta=1$ then the above inequality holds by replacing the term $n^{1-\beta}$ with $\log n$.
\end{corollary}
\begin{proof}
See \citet[Corollary B.1]{toulis2014implicit}.
\end{proof}

\begin{lemma}
\label{lemma:useful}
Suppose Assumptions \aLinearLoss, \aLipZero, and \aFisher\ hold.
Then, almost surely,
\begin{align}
\label{eq:useful:lam}
s_n  & \ge \frac{1}{1 + \gamma_n \maxF},\\
\label{eq:useful:mse2}
||\theta_{n}-\theta_{n-1}||^2 & \le 4\constLipZero^2 \gamma_n^2,
\end{align}
where $s_n$ is defined in Lemma \ref{lemma:implicit_algo},
and $\theta_{n}$ is the $n$th iterate of implicit SGD \eqref{eq:aisgd_implicit}.
\end{lemma}
\begin{proof}
See \citet[Lemma B.3]{toulis2014implicit}.
\end{proof}

\TheoremMSE\
\begin{proof}
See \citet[Theorem 3.1]{toulis2014implicit}.
 \end{proof}

\remark\ \#1.
Assuming Lipschitz continuity of the gradient $\nabla L$ instead
of function $L$, i.e., \assumeLipOne\ over \assumeLipZero\, would not alter the main result
of Theorem \ref{theorem:mse2} about the $\bigO{n^{-\gamma}}$
rate of the mean-squared error.
Assuming Lipschitz continuity with constant $\constLipOne$ of $\nabla L$ and boundedness of $\Ex{||\nabla L(\thetastar, \xi_n)||^2} \le \sigma^2$,
as it is typical in the literature,
would simply add a term $\gamma_n^2 \constLipOne^2 \Ex{||\thetaim{n}-\thetastar||^2} + \gamma_n^2 \sigma^2$ in the corresponding recursive inequality.
Specifically, by Lemma \ref{lemma:implicit_algo},
$s_n\le 1$, and thus
\begin{align}
\Ex{||\nabla L(\theta_n, \xi_n)||^2} = \Ex{s_n^2 ||\nabla L(\theta_{n-1}, \xi_n)||^2} & \le  \Ex{ ||\nabla L(\theta_{n-1}, \xi_n)||^2}
\nn\\
& = \Ex{||\nabla L(\theta_{n-1}, \xi_n) - \nabla L(\thetastar, \xi_n) + \nabla L(\thetastar, \xi_n)||^2}
\nn\\
& \le \constLipOne^2 \Ex{||\theta_{n-1}-\thetastar||^2} + \gamma_n^2 \Ex{||\nabla L(\thetastar, \xi_n)||^2}\nn\\
& \le \constLipOne^2 \Ex{||\theta_{n-1}-\thetastar||^2} + \gamma_n^2 \sigma^2.
\end{align}
The recursion for the implicit errors would then be
\begin{align}
\Ex{||\thetaim{n}-\thetastar||^2} \le (\frac{1}{1 + \gamma_n \minEigF \epsilon} + \constLipOne^2 \gamma_n^2)  \Ex{||\thetaim{n-1}-\thetastar||^2} + \gamma_n^2 \sigma^2\nn,
\end{align}
which also implies the $\bigO{n^{-\gamma}}$ convergence rate.
However, it is an open problem whether it is possible to derive a
nice stability property for implicit SGD under \assumeLipOne\,
similar to the result of Theorem \ref{theorem:mse2} under
\assumeLipZero.

\remark\ \#2. An assumption of almost-sure convexity can
simplify the analysis significantly.
For example, similar to the assumption of \citet{ryustochastic}, assume that $L(\theta, \xi)$ is convex almost surely such that
\begin{align}
(\theta_n-\thetastar)^\intercal \nabla L(\theta_n, \xi_n) \ge
\frac{\mu_n}{2}||\thetaim{n}-\thetastar||^2,
\end{align}
where $\mu_n\ge0$ and $\Ex{\mu_n} = \mu>0$. Then,
\begin{align}
\theta_n + 2\gamma_n \nabla L(\theta_n, \xi_n) & = \theta_{n-1}
\commentEq{by definition of implicit SGD \eqIsgd}\nn\\
||\theta_n-\thetastar||^2 + 2\gamma_n (\theta_n-\thetastar)^\intercal \nabla L(\theta_n, \xi_n) & \le ||\theta_{n-1}-\thetastar||^2.
\nn\\
(1+\gamma_n \mu_n) ||\theta_n-\thetastar||^2 & \le ||\theta_{n-1}-\thetastar||^2.
\nn\\
\Ex{||\theta_n-\thetastar||^2} & \le \frac{1}{1+\gamma_n\mu}\Ex{||\theta_{n-1}-\thetastar||^2}
+ \text{SD}(1+\gamma_n\mu_n)\text{SD}(||\theta_n-\thetastar||^2),
\end{align}
where the last inequality
follows from the identity $\Ex{X Y} \ge \Ex{X}\Ex{Y} - \text{SD}(X)\text{SD}(Y)$. However, $\text{SD}(1+\gamma_n \mu_n) = \bigO{\gamma_n}$, and assuming bounded $\theta_n$
we get
\begin{align}
\Ex{||\theta_n-\thetastar||^2} & \le \frac{1}{1+\gamma_n\mu}\Ex{||\theta_{n-1}-\thetastar||^2}
+ \bigO{\gamma_n},
\end{align}
which indicates a fast convergence towards $\thetastar$. It is also possible
to work with the recursion
\begin{align}
||\theta_n-\thetastar||^2 & \le \frac{1}{1+\gamma_n\mu_n}||\theta_{n-1}-\thetastar||^2,
\end{align}
and then use a stochastic version of Lemma \ref{lemma:implicit_recursion} although the analysis would be more complex in this case.

\section{Proof of Theorem \ref{theorem:mse2_aisgd}}

In this section, we prove Theorem \ref{theorem:mse2_aisgd}.
To do so, we need bounds for $\Ex{||\theta_n-\thetastar||^2}$,
which are available through Theorem \ref{theorem:mse2},
but also bounds for $\Ex{||\theta_n-\thetastar||^4}$, which are
established in the following lemma.

\TheoremMSEfourth\

\begin{proof}
Define
$W_n \triangleq s_n (\thetaim{n-1}-\thetastar)^\intercal \nabla L(\theta_{n-1}, \xi_n)$ for compactness, and proceed as folllows,
\begin{align}
\label{mse4:eq1}
||\thetaim{n}-\thetastar||^2  & =
	||\thetaim{n-1}-\thetastar||^2  - 2 \gamma_n s_n (\thetaim{n-1}-\thetastar)^\intercal \nabla L(\theta_{n-1}, \xi_n)  + \gamma_n^2 ||\nabla L(\theta_n, \xi_n)||^2
\nn\\
||\thetaim{n}-\thetastar||^2  & =
	||\thetaim{n-1}-\thetastar||^2  - 2 \gamma_n W_n + \gamma_n^2 ||\nabla L(\theta_n, \xi_n)||^2
\commentEq{by definition}\nn\\
||\thetaim{n}-\thetastar||^2   & \le
	||\thetaim{n-1}-\thetastar||^2  - 2 \gamma_n W_n  +
	4 \constLipZero^2 \gamma_n^2,
\nn\\
||\thetaim{n}-\thetastar||^4 & \le
||\thetaim{n-1}-\thetastar||^4 + 4\gamma_n^2 W_n^2
+ 16 \constLipZero^4\gamma_n^4\nn\\
& -2\gamma_n ||\theta_{n-1}-\thetastar||^2 W_n + 4\constLipZero^2 \gamma_n^2 ||\theta_{n-1}-\thetastar||^2 - 8\constLipZero^2 \gamma_n^3 W_n.
\end{align}

By Lemma 4 we have
\begin{align}
\label{mse4:eq2}
\ExCond{W_n}{\Fn{n-1}} \ge \frac{\minEigF}{2(1+\gamma_n\maxF)}||\theta_{n-1}-\thetastar||^2.
\end{align}

Furthermore,
\begin{align}
\label{mse4:eq3}
\ExCond{W_n^2}{\Fn{n-1}} & \defeq
\ExCond{[s_n(\theta_{n-1}-\thetastar)^\intercal \nabla \loss(\theta_{n-1}, \xi_n)]^2}{\Fn{n-1}}\nn\\
& \defeq \ExCond{[(\theta_{n-1}-\thetastar)^\intercal \nabla \loss(\theta_{n}, \xi_n)]^2}{\Fn{n-1}}
\commentEq{by Lemma 1}\nn\\
& \le ||\theta_{n-1}-\thetastar||^2 \ExCond{ ||\nabla \loss(\theta_n, \xi_n)||^2}{\Fn{n-1}}
\commentEq{by Cauchy-Schwartz inequality}\nn\\
& \le 4\constLipZero^2 ||\theta_{n-1}-\thetastar||^2
\commentEq{by Lemma 4}
\end{align}

Define $B_n\triangleq\Ex{||\theta_{n}-\thetastar||^2}$ for notational brevity.
We use results \eqref{mse4:eq2} and \eqref{mse4:eq3}
to get
\begin{align}
\label{mse4:eq4}
\Ex{||\theta_{n}-\thetastar||^4} & \le \left(1-\frac{\gamma_n\minEigF}{1+\gamma_n\maxF}\right) \Ex{||\theta_{n-1}-\thetastar||^4}
+ 4\constLipZero^2\gamma_n^2 (5-\frac{\gamma_n\minEigF}{1+\gamma_n\maxF})B_{n-1} + 16\constLipZero^4\gamma_n^4\nn\\
\Ex{||\theta_{n}-\thetastar||^4} & \le \left(1-\frac{\gamma_n\minEigF}{1+\gamma_n\maxF}\right) \Ex{||\theta_{n-1}-\thetastar||^4}
+ 20\constLipZero^2\gamma_n^2 B_{n-1}+ 16\constLipZero^4\gamma_n^4\nn\\
\Ex{||\theta_{n}-\thetastar||^4} & \le \frac{1}{1+\gamma_n \minEigF\epsilon}  \Ex{||\theta_{n-1}-\thetastar||^4}
+ 20\constLipZero^2\gamma_n^2 B_{n-1}+ 16\constLipZero^4\gamma_n^4.
\commentEq{by \assumeFisher}\nn\\
\Ex{||\theta_{n}-\thetastar||^4} & \le \frac{1}{1+\gamma_n \minEigF\epsilon}  \Ex{||\theta_{n-1}-\thetastar||^4}
+ K_0 \gamma_n^3 + e^{-\log \lambda \cdot n^{1-\gamma}} K_1 + K_2\gamma_n^4,
\commentEq{by Theorem \ref{theorem:mse2}}
\end{align}
where $\lambda = (1+\gamma_1(\maxF-\minEigF))^{-1}$
and $\Gamma^2=4\constLipZero^2\sum\gamma_i^2$, (as in Theorem \ref{theorem:mse2}), $K_0\triangleq 160\constLipZero^4\lambda / \minEigF$,
$K_1\triangleq20\constLipZero^2(\Ex{||\theta_0-\thetastar||^2}+ \lambda^{n_0}\Gamma^2)$, and $K_2\triangleq16\constLipZero^4$,
and $n_0$ is a constant defined in the proof of Theorem \ref{theorem:mse2}.

Now, define
\begin{align}
\label{mse4:eq5}
K_3 \triangleq K_0 + K_2\gamma_1 + \max\{\frac{e^{-\log\lambda \cdot \rho_\gamma(n) K_1}}{\gamma_n^3} \},
\end{align}
which exists and is finite. Through simple algebra it is easy to
verify that
\begin{align}
\label{mse4:eq6}
K_0 \gamma_n^3 + e^{-\log \lambda \cdot \rho_\gamma(n)} K_1 + K_2\gamma_n^4 \le K_3 \gamma_n^3,
\end{align}
for all $n$. Therefore, we can simplify Ineq.~\eqref{mse4:eq4} as
\begin{align}
\label{mse4:eq7}
\Ex{||\theta_{n}-\thetastar||^4} & \le \frac{1}{1+\gamma_n \minEigF\epsilon}  \Ex{||\theta_{n-1}-\thetastar||^4} + K_3 \gamma_n^3.
\end{align}
We can now apply Corollary \ref{corollary:implicit_recursion}
with $a_n\equiv K_3\gamma_n^3$ and $b_n \equiv \gamma_n \minEigF \epsilon$ to derive the final bounds for $\Ex{||\theta_n-\thetastar||^4}$.
\end{proof}

We now evaluate the mean squared error of the averaged iterates,
$\thetaBar{n}$.
\TheoremMSEaisgd
\begin{proof}
We leverage a result shown for averaged explicit stochastic gradient descent. In particular, it has been shown that the squared error for the averaged iterate satisfies:
\begin{align}
\label{thm4:eq1}
(\Ex{||\thetaBar{n}-\thetastar||^2})^{1/2} \le
& \frac{1}{\sqrt{n}} \left(\mathrm{trace}(\nabla^2 \exloss(\thetastar)^{-1} \Sigma \nabla^2 \exloss(\thetastar)^{-1})\right)^{1/2} \nn \\
& +  \cTwo (\Ex{||\theta_n-\thetastar||^2})^{1/2} \nn \\
& + \cFour \sum_{i=1}^n (\Ex{||\theta_i-\thetastar||^4}^{1/2}.
\end{align}
The proof technique for \eqref{thm4:eq1} was first
devised by \citet{polyak1992acceleration},
but was later refined by \citet{xu2011towards}, and \citet{moulines2011non}.
In this paper,we follow the formulation
of \citet[Theorem 3, page 20]{moulines2011non};  the derivation of Ineq.\eqref{thm4:eq1}
 for the implicit procedure
is identical to the derivation for the explicit one, however
the two procedures differ in the terms that
appear in the bound \eqref{thm4:eq1}.

All such terms in \eqref{thm4:eq1} have been bounded
in the previous sections. In particular, we can use
Theorem \ref{theorem:mse2} for $\Ex{||\theta_n-\thetastar||^2}$; we can also use Theorem \ref{theorem:mse2_aisgd}
and the concavity of the square-root to derive
\begin{align}
\label{thm4:eq2}
\sum_{i=1}^n (\Ex{||\theta_i-\thetastar||^4}^{1/2}
& \le \sum_{i=1}^n
\left(
\Kone i^{-\gamma} + e^{-\log \lambda \cdot i^{1-\gamma}/2}\Ktwo
 \right) \nn \\
& \le \Kone  n^{1-\gamma} + K_2(n)\Ktwo,
\end{align}
where $K_2(n) = \sum_{i=1}^n\exp\left(- \frac{\log \lambda}{2} i^{1-\gamma}\right)$, $\zeta_0 = \Ex{||\theta_0 - \thetastar||^4}$, and $\Delta^3, n_{0,2}$ are defined
in Lemma~\ref{lemma:mse4}, substituting $n_0$ for $n_{0,2}$.
Similarly, using Theorem \ref{theorem:mse2},
\begin{align}
(\Ex{||\theta_n-\thetastar||^2}^{1/2} \le
\Kthree n^{-\gamma/2} + e^{-\log \lambda \cdot n^{1-\gamma}/2}\Kfour, \nn
\end{align}
where $\delta_0 = \Ex{||\theta_n- \thetastar||^2}$,
and $n_{0,1}, \Gamma^2$ are defined in Theorem \ref{theorem:mse2},
substituing $n_{0,1}$ for $n_0$.
These two bounds can be used in Ineq.\eqref{thm4:eq1}
and thus yield the result of Theorem \ref{theorem:mse2_aisgd}.
\end{proof}

\section{Data sets used in experiments}
\begin{table*}[!h]
  \centering
\begin{tabular}{|l|l|l|l|l|l|l|l|l|l|l|}
\hline
       & description          & type   & features& training set & test set & $\lambda$\\
\hline
covtype& forest cover type    & sparse & 54      & 464,809 & 116,203 & $10^{-6}$\\
delta  & synthetic data       & dense  & 500     & 450,000 & 50,000     & $10^{-2}$\\
rcv1   & text data            & sparse & 47,152  & 781,265 & 23,149  & $10^{-5}$\\
mnist  & digit image features & dense  & 784     & 60,000  & 10,000  & $10^{-3}$\\
sido   & molecular activity   & dense  & 4,932   & 10,142  & 2,536   &
$10^{-3}$\\
alpha  & synthetic data       & dense  & 500     & 400k    & 50k     & $10^{-5}$\\
beta   & synthetic data       & dense  & 500     & 400k    & 50k     & $10^{-4}$\\
gamma  & synthetic data       & dense  & 500     & 400k    & 50k     & $10^{-3}$\\
epsilon& synthetic data       & dense  & 2000    & 400k    & 50k     & $10^{-5}$\\
zeta   & synthetic data       & dense  & 2000    & 400k    & 50k     & $10^{-5}$\\
fd     & character image      & dense  &   900   & 1000k   & 470k    & $10^{-5}$\\
ocr    & character image      & dense  &  1156   & 1000k   & 500k    & $10^{-5}$\\
dna    & DNA sequence         & sparse &   800   & 1000k   & 1000k   & $10^{-3}$\\
\hline
\end{tabular}
\caption{\label{datasets} Summary of data sets and the $L_2$ regularization parameter
$\lambda$ used}
\end{table*}

Table \ref{datasets} includes a full summary of all data sets considered in our
experiments. The majority of regularization parameters are set according to
\citet{xu2011towards}.

\end{document}